\documentclass[letterpaper,10pt,conference]{IEEEtran}
\IEEEoverridecommandlockouts

\usepackage[
    letterpaper,
    left=0.75in,
    right=0.75in,
    top=0.75in,
    bottom=0.75in
]{geometry}

\usepackage{amsmath,amssymb,amsfonts}
\usepackage{cite}
\usepackage{algorithm}
\usepackage{algpseudocode}
\usepackage{graphicx}
\usepackage{textcomp}
\usepackage{xcolor}
\usepackage{lipsum}
\usepackage{dsfont}
\usepackage{calligra}
\usepackage{mathrsfs}
\usepackage{hyperref}
\hypersetup{colorlinks=true,linkcolor=blue,citecolor=blue, linktocpage}

\newtheorem{theorem}{Theorem}[section]

\newtheorem{proposition}[theorem]{Proposition}
\newtheorem{definition}{Definition}[section]
\newtheorem{assumption}{Assumption}
\newtheorem{problem}{Problem}

\usepackage{etoolbox}

\AtEndEnvironment{theorem}{\hfill$\IEEEQED$}
\AtEndEnvironment{corollary}{\hfill$\IEEEQED$}
\AtEndEnvironment{lemma}{\hfill$\IEEEQEDe$}
\AtEndEnvironment{proposition}{\hfill$\IEEEQED$}
\AtEndEnvironment{assumption}{\hfill$\IEEEQED$}

\newcommand{\cJ}{\mathcal{J}}
\newcommand{\cT}{\mathcal{T}}
\newcommand{\cB}{\mathcal{B}}

\newcommand{\norm}[1]{|#1|}
\newcommand{\boresight}{\hat{b}}
\newcommand{\los}{\hat{\ell}}

\def\BibTeX{{\rm B\kern-.05em{\sc i\kern-.025em b}\kern-.08em
    T\kern-.1667em\lower.7ex\hbox{E}\kern-.125emX}}
\begin{document}

\title{Robust Game-theoretic Motion Planning over Extended Time Horizons
\thanks{Outland would like to thank the Department of War Science, Math, and Research for Transformation (SMART) Scholarship for academic funding.}
}

\author{\IEEEauthorblockN{Bennet Outland and Vishala Arya}
\IEEEauthorblockA{\textit{Smead Aerospace Engineering and Sciences Department} \\
\textit{University of Colorado Boulder}\\
Boulder, USA \\
bennet.outland@colorado.edu, vishala.arya@colorado.edu}

\vspace{-0.4in}

}

\maketitle

\begin{abstract}

This work presents a solution to nonconvex, game-theoretic motion planning problems subject to disturbances over long time horizons. The problem is posed as a partially-decoupled generalized Nash equilibrium problem, in which each agent's dynamics depend only on its own state and control, admitting fast solution methods for competitive multi-agent motion planning. An algorithm, WOLF, is developed that applies receding-horizon model predictive control to an open-loop differential games solver based on sequential convexification. In contrast to robust formulations that fix the uncertainty description offline, the robustness tube here is itself a dynamic state, co-optimized with the trajectory, and its thickness directly sets the tightening of the shared coupling constraints. A sufficient condition is derived under which a nominal trajectory satisfying constraints tightened against all agents' error bounds remains feasible for every admissible disturbance realization. The method is demonstrated on two adversarial on-orbit games with coupled translational–attitude dynamics: a stealthy co-orbital jamming game under an active detection-probability bound, and a sun-blocking game in which an adversary disables an evader by decreasing solar power. 
\end{abstract}


\section{Introduction}

Robotic motion planning within the broader robotics research domain has been generally addressed by sampling-based \cite{karaman2011sampling} and sequential convex programming \cite{acikmese2007convex} methods. There are also a growing number of game-theoretic motion planning algorithms for solving multi-agent systems without the need of training ahead of time. The first of this nature is the iterative Linear-Quadratic games (iLQGames) algorithm that is the game theoretic equivalent of the iterated Linear-Quadratic regulator \cite{fridovich2020efficient}. While very fast, it natively lacks the ability to have constraints and does not have convergence guarantees \cite{le2022algames}. Similarly, the Differential Games Sequential Quadratic Programming (DG-SQP) solved quadratic approximations of the KKT conditions which led to a significant increase in reliability and convergence for nonconvex differential games, but at the cost of increased computational time \cite{zhu2024sequential}. More recently, the Fast Augmented Lagrangian Convexification for Open-loop Nash equilibria (FALCON) algorithm was introduced for constrained, nonconvex differential games \cite{outland2026fastconvergentalgorithmsolving}. By reformulating the original game through a sequence of convexified potential-game subproblems, FALCON enables rapid computation of open-loop Nash equilibria while providing stronger convergence properties than previous approaches.

Existing game-theoretic trajectory optimizers are largely formulated using nominal dynamics and therefore do not explicitly account for disturbances arising from environmental effects, modeling errors, or unmodeled dynamics \cite{fridovich2020efficient, le2022algames, zhu2023sequential}. Related work in distributed and robust multi-agent control has demonstrated strong results for cooperative systems, but commonly considers disturbed linear dynamics, assumes potential-game structure, or relies on precomputed worst-case reachable sets \cite{xue2023distributed, sun2024convex, huang2026fast, chen2021scalable, zhan2025robustly}. These assumptions restrict their applicability to nonlinear, high-dimensional systems operating in competitive or adversarial settings. Furthermore, many existing game-theoretic planners compute open-loop solutions over a fixed horizon, making their direct application over extended mission durations increasingly susceptible to model mismatch and computational growth.

The primary contribution of this work is a robust game-theoretic motion-planning framework for high-dimensional nonlinear multi-agent systems with nonconvex constraints, long horizons, and dynamic uncertainty, applicable to both cooperative and competitive scenarios. The method extends FALCON using a receding-horizon formulation \cite{kwon2005receding} and tube-based MPC \cite{mayne2011tube}; specifically the dynamic-tube formulation of \cite{lopez2019dynamic}, in which the tube width is a state driven by an auxiliary control and is not fixed at design time. This yields a sufficient condition under which a nominal trajectory satisfying constraints tightened against all agents' error bounds remains feasible for every bounded admissible disturbance realization.
Secondary contributions include the realistic treatment of two competitive space-robotics test cases: an on-orbit  jamming game and a sun-blocking game. 


\section{Preliminaries} \label{sec: Preliminaries}

\subsection{Notation}

We define the set of indices that correspond to each agent as $[N] = \{1, ..., N \}$ where $i \in [N]$ denotes a given agent. Furthermore, we define $-i := \forall j \in [N] \neq i$ that corresponds to all agents excluding agent $i$. Furthermore, subscripts are used to denote the timestep, $(\cdot)_k$. The time horizon, $T$, is an integer corresponding to the number of timesteps. We then represent the states and control of an agent over a time horizon, defining a trajectory, as $x_k^{(\cdot)} \in X^{(\cdot)} \subseteq \mathbb{R}^n; \, \mathcal{X}^i \in \mathbb{R}^{n \times T}$ and $u_k^{(\cdot)} \in U^{(\cdot)}  \subseteq  \mathbb{R}^m; \, \Upsilon^i \in \mathbb{R}^{m \times T}$, respectively. For notational cleanliness, the following short hand will be occasionally utilized: $\psi^i \equiv (\mathcal{X}^i, \Upsilon^i)$ and $\psi \equiv (\mathcal{X}, \Upsilon)$. The skew-symmetric cross-product will be denoted by $(\cdot)^\times$.

\subsection{Spacecraft Dynamics} \label{subsec: Spacecraft Dynamics}


We consider the rigid-body motion and control of a near-Earth spacecraft. Translational motion is modeled by the non-linear relative motion equations between two spacecraft \cite{schaub2003analytical}. The rigid-body attitude in the $SO(3)$ group is parameterized by modified Rodrigues parameters (MRPs), $\sigma^i \in \mathbb{R}^3$, and the singularity can be avoided by switching to a shadow set \cite{schaub2003analytical}. To avoid the discontinuity in the MRP parameterization, the following constraint is enforced: $\| \sigma^i \| \leq 1$. The kinematic equations are 
\begin{subequations}
    \begin{equation}
        \dot{\sigma}^i = \frac{1}{4} B(\sigma^i) \omega^i, 
    \end{equation}
    \begin{equation}
        B(\sigma^i) = \left( 1 - (\sigma^i)^T \sigma^i \right)I+ 2 [(\sigma^i)^\times] + 2\sigma^i (\sigma^i)^T
    \end{equation}
\end{subequations}
\noindent where $\omega^i$ is the body rate. The rotational equation of motion is then,
\begin{equation}
    J^i \dot{\omega}^i = - [(\omega^i)^\times] J^i \omega^i + \tau^i, 
\end{equation}
\noindent where $J^i$ is the inertia tensor and $\tau^i$ is reaction wheels torque. 

\subsection{Fast Augmented Lagrangian Convexification for Open-loop Nash equilibria (FALCON)}

FALCON is a game-theoretic trajectory planning algorithm that takes inspiration from the Scvx* algorithm \cite{oguri2023successive}. FALCON decomposes the difficult non-convex differential game into convex potential games as seen in Algorithm \ref{alg: FALCON}. This allows for solutions to nonconvex differential games with continuous time constraints. Given mild assumptions that are common in control theory:

\begin{assumption}[Properties of System Dynamics] 
    \label{assump:system-properties}
    The system dynamics satisfy the following properties:
    \begin{enumerate}
        \item Bounded acceleration: 
        The second time derivative of state trajectories is uniformly bounded:
        \begin{equation}
            \|\ddot{x}^i(t)\| \leq M_2 \in \mathbb{R}_+, \quad \forall t \geq 0.
        \end{equation}
        \item Lipschitz continuity: 
        The dynamics $f$ are Lipschitz continuous in both arguments. That is, there exist constants $K_u, K_x > 0$ such that for all $x_0, x_1 \in \mathbb{X}$ and $u_0, u_1 \in \mathbb{U}$:
        \begin{align}
            \|f(x_0^i, u_0^i) - f(x_0^i, u_1^i)\| &\leq K_u \|u_0^i - u_1^i\|, \label{eq:lipschitz-control} \\
            \|f(x_0^i, u_0^i) - f(x_1^i, u_0^i)\| &\leq K_x \|x_0^i - x_1^i\|. \label{eq:lipschitz-state}
        \end{align}
    \end{enumerate}
\end{assumption}

\begin{assumption} \label{assump: cost}
    The cost functional and constraint functions are assumed to have the following properties:
    \begin{enumerate}
        \item $J^i(\mathcal{X}, \Upsilon) \in \mathcal{C}^2$, \\ $C(\mathcal{X}, \Upsilon) \in \mathcal{C}^2$;
        \item $|| \nabla_{\mathcal{X}, \Upsilon} J^i(\mathcal{X}, \Upsilon)|| \leq G_{\rm max} < \infty$, \\ $|| \nabla_{\mathcal{X}, \Upsilon^i} C(\mathcal{X}, \Upsilon)|| \leq G_{\rm max} < \infty$;
        \item $|| \nabla^2_{\mathcal{X}, \Upsilon} J^i(\mathcal{X}, \Upsilon)|| \leq H_{\rm max} < \infty$, \\$|| \nabla^2_{\mathcal{X}, \Upsilon} C(\mathcal{X}, \Upsilon^i)|| \leq H_{\rm max} < \infty$.
    \end{enumerate}
\end{assumption}


\begin{theorem}[Global Strong Convergence with Feasibility] \label{thm: main convergence}
    Under Assumptions \ref{assump: cost} and \ref{assump:system-properties}, and provided the second-order sufficient conditions of \cite[Assumption 1]{oguri2023successive} hold at the limit, Algorithm \ref{alg: FALCON} converges globally to a point, $\psi^*$, within its feasible region that is:
   i) a local variational Nash equilibrium of the original non-convex differential game;
   ii) dynamically feasible; iii) approximately coupling-feasible in the continuous-time sense for every $t \in [t_k, t_{k+1}]$ and every coupling constraint component.
\end{theorem}


\begin{IEEEproof}
    See \cite{outland2026fastconvergentalgorithmsolving}.
\end{IEEEproof}

\begin{algorithm}[t]
\caption{FALCON: Fast Augmented Lagrangian Convexification for Open-loop Nash Equilibria}
\label{alg: FALCON}
\begin{algorithmic}[1]
 
\State INITIALIZE\_VARIABLES$()$ \Comment{Duals and Penalty}
 
\Repeat 
 
    \For{$ i \in [N]$}
        \State \textsc{LOCAL\_LINEARIZATION}$()$ \Comment{Problem \ref{prob: general formulation}}
    \EndFor
    \State ASSEMBLE\_POTENTIAL() \Comment{Local Potential Game}
    \State {SOLVE\_CVX\_SUBPROBLEM}$()$ \Comment{SOCP}
    \State \textsc{SOL\_QUALITY}$()$ \Comment{Check Local Approximation}
    \If{ACCEPT\_STEP} \Comment{Good Approximation}
        \If{IMPROVEMENT} \Comment{Closer to Nash Eq.}
            \State  \textsc{UPDATE\_DUALS}$()$
        \EndIf
        \State \textsc{UPDATE\_TRUST\_REGION}$()$ 
    \Else
        \State REJECT\_STEP() \Comment{Shrink Trust Region}
    \EndIf
\Until{Converged} \Comment{Local Nash Equilibria}
 
\State \Return Trajectories
\end{algorithmic}
\end{algorithm}

\subsection{Dynamic Tube Model Predictive Control}

Recalling the dynamics defined in Section \ref{subsec: Spacecraft Dynamics}, this model is used with disturbance from unmodeled dynamics. 
\begin{assumption}[\cite{lopez2019dynamic}]
    The dynamics $f$ can be expressed as $f = \hat{f} + \tilde{f}$ where $\hat{f}$ is the nominal dynamics and $\tilde{f}$ is the bounded model error (i.e., $|\tilde{f}(x)| \leq \Delta(x)$).
\end{assumption}

\begin{assumption}[\cite{lopez2019dynamic}] \label{asmp: disturbance}
    The disturbance $d^i$ belongs to a closed, bounded, and connected set $\tilde{D}$ (i.e., $\mathbb{D} := \{ d^i \in \mathbb{R}^n: |d^i| \leq \tilde{D}\}$) and is in the span of the control input matrix.
\end{assumption}

\noindent We then consider an ancillary controller per each agent, $\kappa^i$, to keep the state, $x^i$, within a robust control invariant (RCI) tube about, $x^{i, *}$, the per-agent nominal trajectory. 

\begin{definition}[\cite{lopez2019dynamic}] \label{def: tube}
    Let $\mathbb{X}$ denote the set of allowable states and let $\tilde{x}^i := x^i - x^{i, *}$. The set $\Omega^i \subset \mathbb{X} \subset \mathbb{R}^n$ is a RCI tube for an agent, $i$, if there exists an ancillary controller $\kappa^i(x^i, x^{i, *})$ such that if $\tilde{x}^i(t_0) \in \Omega$, then, for all realizations of the disturbance and modeling error, $\tilde{x}^i(t) \in \Omega \, \forall t \geq t_0$.
\end{definition}

Sliding control is employed as the ancillary controller in Definition \ref{def: tube} with a corresponding sliding manifold and sliding variable, $s_j^i$, for the $j$th state. Let $\mathcal{B}_j^i := \{x^i : |s_j^i| \leq \Phi_j^i \}, \forall j\in n$ be the boundary layer on the sliding manifold with thickness $\Phi_j^i$. With a sliding surface convergence rate parameter, $\eta_j^i$, the boundary layer can be attractive. This promotes a robust invariant neighborhood to minimize control signal chatter.

The ancillary controller, $\kappa$, can then be shown to be
\begin{align}
    u^i = & B(x)^{-1} \big[ - \hat{F}^i(x^i) - x_{r, j}^{i} - \\
    & (\alpha^i \Phi^i + \Delta(x^i) - \Delta(x^{i,*})) \text{ sat}(s^i / \Phi^i)\big], \notag
\end{align}
\noindent by using the augmented tube dynamics 
\begin{align}
    \dot{\Phi}^i &= - \alpha^i \Phi^i + \Delta(x^{i, *}) + D^i + \eta^i, \\
    \dot{\alpha}^i &= v^i
\end{align}
\noindent where $\alpha^i$ is a cutoff frequency parameter and $v^i \in \mathbb{V} \in \mathbb{R}^m$ is an artificial control  variable for enforcing tube stability \cite{lopez2019dynamic}.

\section{Disturbed Partially-Decoupled Generalized Nash Equilibrium Problems} \label{sec: PDGNEP}

\subsection{General Formulation}

We consider a mild relaxation of a generalized Nash equilibrium problem (GNEP) for differential games \cite{outland2026fastconvergentalgorithmsolving}. For $i\in[N]$ agents, there are states, $\mathcal{X}^i$, and controls, $\Upsilon^i$. When applied to multi-agent robotics, it can be assumed that the dynamics of each agent, $D^i(\mathcal{X}^i, \Upsilon^i)$, is separable (i.e. the state and control of an agent only applies to the dynamics of that agent). Additionally, each agent has dynamical disturbance, $d^i$, that follows Assumption \ref{asmp: disturbance}. The magnitude of the disturbance can be chosen in accordance with uncertainty in the dynamics model. Under this disturbance, each agent minimizes an extended value cost function $J^i(\mathcal{X}, \Upsilon): \mathbb{R}^{n\times T} \times \mathbb{R}^{m\times T} \rightarrow \mathbb{R} \bigcup \{ +\infty \}$ subject to inequality constraints, $C(\mathcal{X}, \Upsilon) \leq 0$. The general form of the problem can then be defined as

\begin{problem}[Disturbed Partially-Decoupled Generalized Nash Equilibrium Problem (D-PDGNEP)] 
    \begin{equation*} \label{prob: general formulation}
        \begin{alignedat}{2}
        &\min_{\mathcal{X}^i, \Upsilon^i} \quad && J^i(\mathcal{X}, \Upsilon) \\
        \forall i \in [N]  \,\,\,\,\, & \,\text{s.t.} \quad && D^i(\mathcal{X}^i, \Upsilon^i) - d^i = 0, \\
        &&& C(\mathcal{X}, \Upsilon) \leq 0.
        \end{alignedat}
    \end{equation*}
\end{problem}


\begin{definition}[D-PDGNEP Solution] \label{def: solution}
    For a given game, a set of strategies, $\Upsilon^*$, and trajectories, $\mathcal{X}^i$,  subject to every admissible disturbance, $d^i$, are a variational Nash Equilibrium if for every agent $i \in [N]$ with strategy $\Upsilon^i$ such that, $J^i(\mathcal{X}, \Upsilon^{*}) \leq J^i(\mathcal{X}, \{ \Upsilon^{i},\Upsilon^{-i*} \})$ and equivalent dual values for shared constraints, $C(\cdot)$.
\end{definition}

\subsection{Stealthy Co-orbital Jamming Game}

 Consider a non-cooperative game wherein a jammer spacecraft seeks to jam another target satellite on-orbit that is attempting to receive a ground-station command during a pass. Furthermore, the jamming satellite is to perform this operation while being constrained on the detection probability. Following literature on space electronic warfare, we define gains for each spacecraft \cite{adamy2021ew, fotiadis2026optimal}. Let $\boresight^i(\sigma^i)$ denote the antenna boresight for an agent that follows a main-lobe model
\begin{equation}
    g_i(c) = g_{\min}+(1-g_{\min})e^{-\beta_i(1-c)},\\
\end{equation}
where $g_{\min}$ is the normalized sidelobe floor, $\Theta_i$ is the beam width, $\beta_i$ is selected such that $g_i(\cos(\Theta_i/2))=0.5$ for defining half-power beam width. The pointing cosines for some reference $A$ and $B$ are $c_{AB} =\boresight^A\cdot\los_{AB}$ for a unit line-of-sight vector $\los_{AB}$.

The quality of the ground-station signal is measured by the signal-to-interference-noise ratio (SINR),
\begin{subequations}
    \begin{align}
        \mathrm{SINR}
        &=
        \frac{\sigma_r g_\cT(c_{\cT g})}
        {1+\sigma_r
        \left(\dfrac{d_{\rm eq}^2}{d_\epsilon^2}\right)
        g_\cJ(c_{\cJ\cT})g_\cT(c_{\cT\cJ})},\\
        d_{\rm eq}^2
        &=
        \frac{\kappa_\cJ G_\cJ^{\max}G_\cT^{\max}}{\sigma_r}.
    \end{align}
\end{subequations}
In this formulation, $\sigma_r = S_0 / N_0$ is the unjammed ratio, $\kappa_\cJ$ are the jammer transmission and propagation constants, and $d_{eq}$ is the characteristic jamming range. For this game, the SINR is used as a stage cost that the jammer seeks to minimize while the target seeks to maximize. The performance of each agent depends on both positions and attitudes.

Additionally, the jammer needs to reduce the probability that it will be detected during this operation from both a visual \cite{fotiadis2026optimal} and electro-magnetic perspective \cite{adamy2021ew}. The state of the jammer is augmented to add a detection probability state, $\Lambda^\cJ$. We model the detection event as an inhomogeneous Poisson process whose rate depends on the instantaneous relative geometry and thrust magnitude, with the value of
\begin{equation}
    \dot{\Lambda}^\cJ
    =
    \underbrace{
    \frac{k_{\det}\kappa_\cJ G_\cJ^{\max}g_\cJ(c_{\cJ\cT})}
    {d_\epsilon^2}
    }_{\text{RF intercept}}
    +
    \underbrace{
    k_{\rm thr}
    \left(
    \sqrt{\norm{\mathbf a^\cJ}^2+\epsilon_a^2}-\epsilon_a
    \right)
    }_{\text{thruster detection}} .
\end{equation}
The detection probability, $\mathbb{P}_{det}(T)$, is the probability of at least one detection over $[0,T]$. The cumulative probability of detection can be constrained by a upper-bound, $\bar{P}$.
\begin{subequations}
    \begin{align}
        \mathbb P_{\det}(T) &= 1-e^{-\Lambda^\cJ(T)},\\
        \Lambda^\cJ(T)&\leq-\ln(1-\bar P).
    \end{align}
\end{subequations}


\subsection{Sun-Blocking Game}

Energy management is important to complete spacecraft mission objectives \cite{ARYA2026113492}. To exploit this, consider two spacecraft where one attempts to block the sun from shining on a target satellite's solar panels via an extended opaque disk circumscribed on the body. The goal of the blocker satellite is to reduce the battery recovery of the target spacecraft in order to temporarily disable it by casting a shadow on the evader's solar panels \cite{10115968}. Unlike previous work, we directly consider the energy denial as the part of the objective instead of just positioning. 


First, the spacecraft shadowing model is developed \cite{zapevalin2026comparative}. Consider the angular radius of the Sun, $\rho_\odot$, with a solar disk approximated by a fixed quadrature, $\{(\xi_m,w_m)\}_{m=1}^{N_\odot}$, where $\sum_m w_m=1$. For directions $e_1,e_2$ that are orthogonal to the solar disk, a solar ray, $m$, has direction
\begin{equation}
\hat{s}_m
= \frac{
\hat{s}
+
\tan\rho_\odot
\left(
\xi_{1m}e_1+\xi_{2m}e_2
\right)}
{
\norm{
\hat{s}
+
\tan\rho_\odot
\left(
\xi_{1m}e_1+\xi_{2m}e_2
\right)
}_2
},
\end{equation}
where $\hat{s}$ is the central sun normal vector and $\hat n_\cB(\sigma^\cB)$ is the normal vector of the blocker's disk. Considering a point $p$ on the target's solar panel, the intersection of solar ray $\hat{s}_m$ with the blocker is
\begin{subequations}
\begin{align}
\lambda_m^\epsilon
&=
\left[
\hat n_\cB^\top(r_\cB-p)
\right]
\frac{
\hat n_\cB^\top\hat{s}_m
}{
(\hat n_\cB^\top\hat{s}_m)^2+\epsilon_{\rm ray}^2
},\\
x_m^\epsilon
&=
p+\lambda_m^\epsilon\hat{s}_m .
\end{align}
\end{subequations}
The fraction of the $m$th ray blocked by the disk with a radius of $R_\cB$ is smoothly approximated as
\begin{equation}
b_m(p)
=
\operatorname{sig}
\left(
\frac{\lambda_m^\epsilon}{\delta_\lambda}
\right)
\operatorname{sig}
\left(
\frac{
R_\cB^2-\norm{x_m^\epsilon-r_\cB}_2^2
}{
\delta_R^2
}
\right),
\end{equation}
where $\operatorname{sig}(z)=(1+e^{-z})^{-1}$. 
The smoothed model of the occlusion is then modeled as
\begin{equation}
\kappa(p)
=
\sum_{m=1}^{N_\odot}w_m b_m(p).
\end{equation}
In this formulation, $\kappa\approx1$ denotes full occlusion in the umbral cone, $0<\kappa<1$ is a partial blockage (possibly in penumbra), and $\kappa\approx0$ is no-blockage.

The target is modeled as having two solar panel wings of equal size with total area $A_p$ and collecting-area densities $w_+(q)$ and $w_-(q)$. Their power-weighted occultation fractions are
\begin{equation}
\bar{\kappa}_\pm
=
\frac{1}{A_p}
\int_{\mathbb R^2}
w_\pm(q)
\kappa\!\left(
r_\cT+
R_\cT(\sigma^\cT)
\begin{bmatrix}
0\\q
\end{bmatrix}
\right)dq .
\end{equation}
The solar power generated by the target is then
\begin{equation}
P_{\rm gen}
=
\eta S_\odot
h_\epsilon
\left(
\hat n_{\rm arr}^{\top}\hat{s}
\right)
A_p
\left(
2-\bar{\kappa}_+-\bar{\kappa}_-
\right),
\end{equation}
where $\eta$ is the efficiency of the solar cells, $S_\odot$ is the solar irradiance, $\hat n_{\rm arr}=R_\cT(\sigma^\cT)e_x$ is the array normal, and
\begin{equation}
h_\epsilon(c)
=
\frac{1}{2}
\left(
c+\sqrt{c^2+\epsilon_c^2}
\right)
\end{equation}
is a smooth approximation of the positive orthant. Thus, generated power depends on the positions and attitudes of both agents.

The power consumption of the target is modeled by a constant draw from the bus as well as the propulsion system. The net power is
\begin{subequations}
\begin{align}
P_{\rm net}
&=
P_{\rm gen}-P_{\rm bus}-P_{\rm thr},\\
P_{\rm thr}
&=
P_{\rm thr,max}
\frac{
\norm{F^\cT}_2
}{
F_{\max,\cT}
}.
\end{align}
\end{subequations}
The net power is used as the stage cost of the game where the blocker seeks to decrease it to disable the target, while the target seeks to maximize it. The target then needs to balance between escaping a shadowed region and using thrust that further decreases battery capacity. If the battery reaches a low-power state, determined by integrating $P_{net}$, the Target enters a safe-mode and is functionally disabled.


\section{Proposed Solution} \label{sec: Proposed Solution}

\subsection{Constraint Tightening} \label{subsec: Constraint Tightening}
We seek to tighten the general inequality constraints of Problem \ref{prob: general formulation} against bounded dynamical disturbances, where robustness is provided by an ancillary controller subject to the available control-authority constraints.

\begin{proposition} \label{prop: box}
Let $\gamma^i_j = \Phi_j^i(t) / \alpha^i(t)$ be a $j$ component-wise bound on the disturbed trajectory \cite{slotine1991applied}. Consider a box such that $\mathcal{T} = \{|\tilde{\psi}| \leq \gamma \}: \sup_{\tilde{\psi} \in \mathcal{T}} g^T \tilde{\psi} = |g|^T \gamma$ and $\sup_{\tilde{\psi} \in \mathcal{T}} \| \tilde{\psi} \|^2 = \| \gamma \|^2$.
\end{proposition}
\begin{theorem}
    Let $C(\psi) \leq 0$ satisfy Assumption \ref{assump: cost}. Recall, $\psi \equiv(\mathcal{X}, \Upsilon)$. Let the disturbance state of each agent lie in the box $\mathcal{T}^i$. Robust feasibility is guaranteed if, 
    \begin{equation*}
        C(\psi^*) + \sum_{i=1}^N | \nabla_{\psi^i} C(\psi^*)|^T \gamma^i + \frac{H_{max}}{2} \sum_{i=1}^N \| \gamma^i \|^2 \leq 0,
    \end{equation*}
    then $C(\psi^* + \tilde{\psi}) \leq 0$ for every admissible $\tilde{\psi} \in \Pi_i \mathcal{T}^i$.
\end{theorem}
\begin{IEEEproof}
    We seek to determine an upper-bound for the expression $\sup_{\tilde{\psi} \in \mathcal{T}} C(\psi^* + \tilde{\psi}) \leq 0$, where $\mathcal{T} = \Pi_i \mathcal{T}^i$. Consider a trajectory disturbance, $\tilde{\psi} \in \mathcal{T}$. A Taylor series expansion is taken on $\psi^*$ to $\psi^* + \tilde{\psi}$ with a second-order remainder. Bounded from above using Assumption \ref{assump: cost},
    \begin{equation*}
        C(\psi^* + \tilde{\psi}) \leq C(\psi^*) + \nabla C(\psi^*)^T \tilde{\psi}+ \frac{1}{2}H_{max}\| \tilde{\psi} \|^2.
    \end{equation*}
    The supremum over the box in Proposition \ref{prop: box} gives,
    \begin{equation*}
       \sup_{\tilde{\psi} \in \mathcal{T}} C(\psi^* + \tilde{\psi}) \leq C(\psi^*) + \sup_{\tilde{\psi} \in \mathcal{T}} \nabla C(\psi^*)^T \tilde{\psi}+ \sup_{\tilde{\psi} \in \mathcal{T}}\frac{1}{2}H_{max}\| \tilde{\psi} \|^2.
    \end{equation*}
    Applying Proposition \ref{prop: box} to the linear and quadratic terms on the right hand side results in a sufficient condition for robust constraint satisfaction,
    \begin{equation*}
        C(\psi^*) + \sum_{i=1}^N | \nabla_{\psi^i} C(\psi^*)|^T \gamma^i + \frac{H_{max}}{2} \sum_{i=1}^N \| \gamma^i \|^2 \leq 0.
    \end{equation*}
\end{IEEEproof}

\subsection{Worst-case, Online, Long-horizon FALCON (WOLF)}

With the constraint tightening defined in Section \ref{subsec: Constraint Tightening}, we can define a modified form of Problem \ref{prob: general formulation} to dynamically define the RCI tube. Additionally, to aid in adapting modeling error and longer time horizons, a receding horizon formulation is used \cite{mattingley2011receding}. 

\begin{problem}[Receding-Horizon, Dynamic-Tube PDGNEP] 
    \begin{equation*} \label{prob: DT formulation}
        \begin{alignedat}{2}
        &\min_{\mathcal{X}^i, \Upsilon^i, v^i} \quad && \frac{1}{T_H + 1} \sum_{\tau=t}^{t+T_H}  J^i_{\tau | t}(\mathcal{X}_\tau, \Upsilon_\tau) \\
        & \,\text{s.t.} \quad && D^i(\mathcal{X}^i, \Upsilon^i) - d^i = 0, \forall i \in [N] ,\\
        &&& \dot{\Phi}^i_\tau = - \alpha^i \Phi^i_\tau + \Delta^i(x^{i, *}_\tau) + d^i + \eta^i, \\
        &&& \dot{\alpha}^i = v^i, \alpha^i \in [\underline{\alpha}, \overline{\alpha}^i], |v^i| \leq \overline{v}, \\
        &&& C_{\tau | t}(\mathcal{X}^*_\tau, \Upsilon^*_\tau) + \sum_{i=1}^N | \nabla_{\mathcal{X}_\tau^i, \Upsilon^i_\tau} C_{\tau | t}(\mathcal{X}^*_\tau, \Upsilon^*_\tau)|^T \gamma^i\\ &&& \:\: \: \: \: \: \: \: \: \: \: \: \: \: \: \: \: \: \: \: \:  + \frac{H_{max}}{2} \sum_{i=1}^N \| \gamma^i \|^2 \leq 0,
        \end{alignedat}
    \end{equation*}
\end{problem}

\noindent For this problem, the user set parameters not already defined are as follows: $\alpha^i$, $\eta^i$, upper and lower bounds on $\alpha$ $([\underline{\alpha}, \overline{\alpha}^i])$, and maximum value of $v$ $(\overline{v})$. All other values can be determined from the dynamics, constraints, or disturbance. Problem \ref{prob: DT formulation} can then be solved in a receding horizon manner via Algorithm \ref{alg: FALCON}. We denote this reformulation as WOLF.

\section{Numerical Results} \label{sec: Numerical Results}

Using the WOLF formulation, we solve parameterizations of the Stealthy Co-orbital Jamming Game and Sun-Blocking Game defined in Section \ref{sec: PDGNEP}. Note: Local solutions are discovered and there may be many different Nash equilibria present in the game. The games herein were implemented using the DifferentialGames.jl ecosystem \cite{outland_2026_20432309}. In each game, the disturbance in the dynamics is at most 3\% of the applied control. Across 50 independent disturbance realizations of the jamming game, the constraint violation percentages were 26\% for FALCON, 20\% for FALCON in a receding horizon loop, and 0\% for WOLF with the worst case tube utilization observed at approximately 50\%: which is a normalized measure of tube boundary proximity. This demonstrates the effectiveness of the dynamic tube. Computationally, Algorithm \ref{alg: FALCON} required a median of 11 SCP iterations per receding-horizon 30 second look-ahead step over a simulation horizon of 2850 seconds, with a median solve time of 1.63 seconds per step on a AMD Ryzen 7 7445HS for the sun-blocking game.

\subsection{Stealthy Co-Orbital Jamming Game}

As seen in Figure \ref{fig: jamming}, the jammer starts in a periodic line segment via a near natural motion relative trajectory attempting to jam the ground station signal from the target. The jammer begins jamming at a further distance while the target responds by moving out of plane. During this maneuver, there is high signal degradation as the SINR drops to 1.84 dB (mean 9.08 dB over the horizon). The dominant control employed by the target is attitude re-pointing after the SINR drops below 5 dB at approximately 5 minutes into the game. During this slewing, the mean ground-pointing departure significantly increases ($2.93^\circ \to 6.09^\circ$) in an attempt to mitigate the impact of the jammer. Ultimately, the jammer backs off due to the detection constraint and keeps it relatively loose to maintain robustness to disturbances.

\begin{figure}
    \centering
    \includegraphics[width=\linewidth]{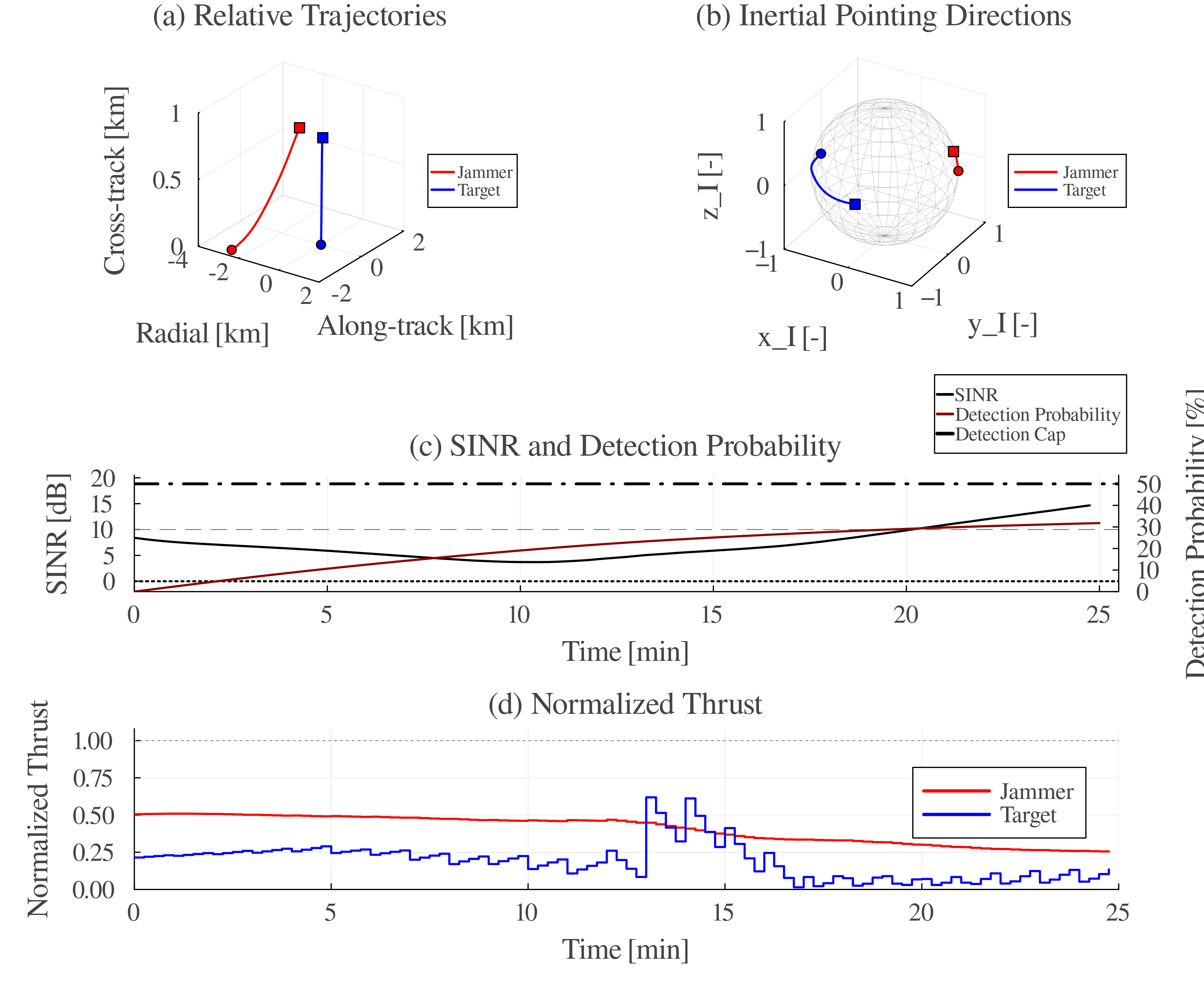}
    \caption{Jamming game: The trajectories in position and pointing are shown in subfigures a and b with the circle denoting the initial conditions and the square denoting the final state. A combined figure of the SINR competition and the detection probability are overlain. Finally, the normalized thrust profiles of both agents are displayed in subfigure d.}
    \label{fig: jamming}
\end{figure}

\subsection{Sun-Blocking Game}

For an example instantiation of the sun-blocking game, as displayed in Figure \ref{fig: blocking}, the blocker occults the right solar panel at the start and quickly follows with the left panel such that the gain quickly becomes 0\%; the evader escapes the initial blockage at about 4 minutes into the game and is re-acquired at about 17 minutes. The evader becomes temporarily blocked again near 21 and 45 minutes into the game. During the interim, the blocker is forcing the evader to make brief slews that cause sun pointing error and decrease the power gain as seen by the red shaded regions in Figure \ref{fig: blocking} near 7, 20, 30, and 40 minutes into the game. During this time, the evader continues loosing battery power as it must continue thrusting while also loosing potential power gain from the solar panels decreasing from 50\% to 8.7\%. Ultimately, the low battery of the evader activates the safe mode for the satellite and it is temporarily disabled.

\begin{figure}
    \centering
    \includegraphics[width=\linewidth]{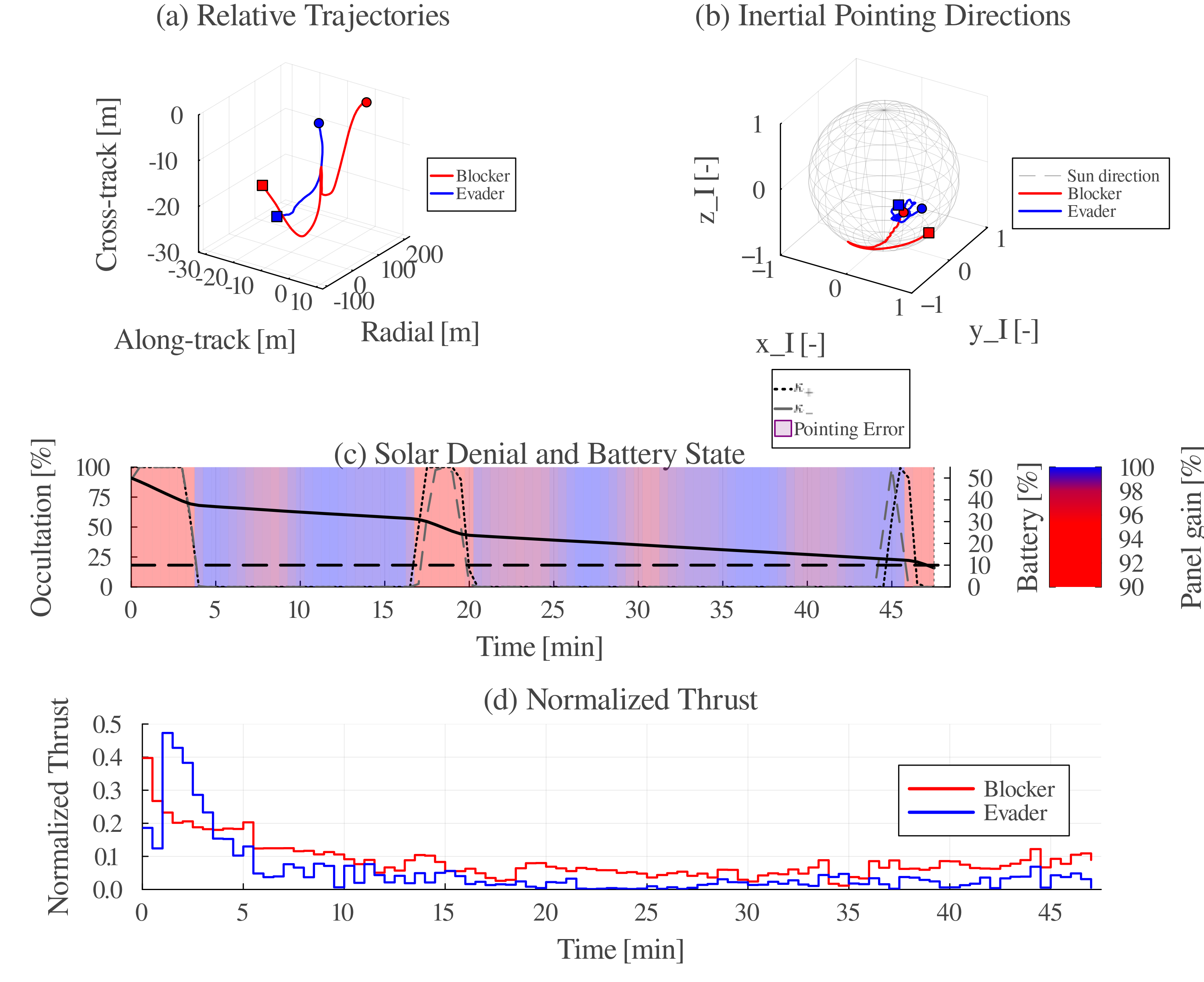}
    \caption{Sun-blocking game: The trajectories in position and pointing are shown in subfigures a and b with the circle denoting the initial conditions and the square denoting the final state. A combined figure of the occlusion fractions of the two solar panels of the evader, denoted by $\kappa_{ ( \cdot) }$ and the evader battery percentage are overlain. The background coloring denotes panel gain variations due to evader pointing errors. Finally, the normalized thrust profiles of both agents are displayed in subfigure d. }
    \label{fig: blocking}
\end{figure}

\section{Conclusion} \label{sec: Conclusion}

In this work, we introduced a novel method, WOLF, for robustly solving multi-agent motion planning problems over long time horizons. Leveraging previous work in robust single agent motion planning, tube model predictive control techniques are extended to multi-agent systems via a game-theoretic formulation. This direction included developing a sufficient condition for robust constraint tightening for generalized multi-agent constraints. The numerical effectiveness of this method was demonstrated via an on-orbit jamming game and a sun-blocking game. The WOLF algorithm extends the capabilities of differential game solvers to be robust to disturbances while solving competitive game scenarios.

\section*{Acknowledgment}
Outland dedicates his contribution {\footnotesize \calligra S.D.G}.

\section*{Generative AI Disclosure}
Generative AI was used for literature review aid, software implementation of the algorithms in this work, and as an informal proof assistant and aid. All work has been independently verified by the authors. Various models from OpenAI and Anthropic were used.

\bibliographystyle{IEEEtran}
\bibliography{ref}

\end{document}